\documentclass{article}
\usepackage{amsmath,amssymb,amsfonts}

\def\op#1{\mathop{{\it\fam0} #1}\limits}

\newcommand{\beq}{\begin{equation}}
\newcommand{\eeq}{\end{equation}}
\newcommand{\ben}{\begin{eqnarray}}
\newcommand{\een}{\end{eqnarray}}
\newcommand{\be}{\begin{eqnarray*}}
\newcommand{\ee}{\end{eqnarray*}}
\newcommand{\bea}{\begin{eqalph}}
\newcommand{\eea}{\end{eqalph}}

\newcommand{\al}{\alpha}
\newcommand{\bt}{\beta}
\newcommand{\dl}{\delta}
\newcommand{\la}{\lambda}

\newcommand{\m}{\mu}

\newcommand{\g}{\gamma}

\newcommand{\vt}{\vartheta}
\newcommand{\cG}{{\mathfrak g}}

\newcommand{\up}{\upsilon}

\newcommand{\si}{\sigma}
\newcommand{\Si}{\Sigma}

\newcommand{\bb}{{\mathbf 1}}

\newcommand{\dr}{\partial}

\newcommand{\ar}{\op\longrightarrow}

\newenvironment{eqalph}{\stepcounter{equation}
\setcounter{equationa}{\value{equation}} \setcounter{equation}{0}

\begin{eqnarray}}{\end{eqnarray}\setcounter{equation}{\value{equationa}}}

\newcounter{equationa}
\newcounter{remark}
\newcounter{example}
\newcounter{theorem}
\newcounter{proposition}
\newcounter{lemma}
\newcounter{corollary}
\newcounter{definition}
\def\theremark{\arabic{remark}}

\def\thedefinition{\arabic{theorem}}

\newenvironment{proof}{{\it Proof.}}{
\medskip }

\newenvironment{theorem}{\refstepcounter{theorem} \medskip{\bf
Theorem \thedefinition.}\it}{\medskip }

\newenvironment{lemma}{\refstepcounter{theorem} \medskip{\bf  Lemma
\thedefinition.}\it }{\medskip }

\newcommand{\mar}[1]{}

\begin{document}

\hbox{}

\begin{center}

{\Large\bf Theory of Classical Higgs Fields. I. Matter Fields}

\bigskip

G. SARDANASHVILY, A. KUROV

\medskip

Department of Theoretical Physics, Moscow State University, Russia

\bigskip

\end{center}

\begin{abstract}
Higgs fields are attributes of classical gauge theory on a
principal bundle $P\to X$ whose structure Lie group $G$ if is
reducible to a closed subgroup $H$. They are represented by
sections of the quotient bundle $P/H\to X$. A problem lies in
description of matter fields with an exact symmetry group $H$.
They are represented by sections of a composite bundle which is
associated to an $H$-principal bundle $P\to P/H$. It is essential
that they admit an action of a gauge group $G$.
\end{abstract}

\bigskip

Higgs fields are attributes of classical gauge theory on a
principal bundle $P\to X$ if its symmetries are spontaneously
broken \cite{book09,sard92,sard06a}. Spontaneous symmetry breaking
is a quantum phenomenon, but it is characterized by a classical
background Higgs field \cite{sard08}. Therefore, such a phenomenon
also is considered in the framework of classical field theory when
a structure Lie group $G$ of a principal bundle $P$ is reduced to
a closed subgroup $H$ of exact symmetries.

One says that a structure Lie group $G$ of a principal bundle $P$
is reduced to its closed subgroup $H$ if the following equivalent
conditions hold:

$\bullet$ a principal bundle $P$ admits a bundle atlas with
$H$-valued transition functions,

$\bullet$ there exists a principal reduced subbundle of $P$ with a
structure group $H$.

\noindent A key point is the following.

\begin{theorem}\label{redsub} \mar{redsub}
There is one-to-one correspondence between the reduced
$H$-principal subbundles $P^h$ of $P$ and the global sections $h$
of the quotient bundle $P/H\to X$ possessing a typical fibre
$G/H$.
\end{theorem}

In classical field theory, global sections of a quotient bundle
$P/H\to X$ are treated as classical Higgs fields
\cite{book09,sard06a}.

In general, there is topological obstruction to reduction of a
structure group of a principal bundle to its subgroup. In
particular, such a reduction occurs if the quotient $G/H$ is
diffeomorphic to a Euclidean space $\mathbb R^m$. For instance,
this condition is satisfied if $H$ is a maximal compact subgroup
$H$ of a Lie group $G$ \cite{ste}.

Given a principal bundle $P\to X$ whose structure group is
reducible to a closed subgroup $H$, one meets a problem of
description of matter fields which admit only an exact symmetry
subgroup $H$. Here, we aim to show that they adequately are
represented by sections of the composite bundle $Y$ (\ref{b3225})
which is associated to an $H$-principal bundle $P\to P/H$. A key
point is that $Y$ also is a $P$-associated bundle (Theorem
\ref{k12}), and it admits the action (\ref{rty}) of a gauge group
$G$. In the case of a pseudo-orthogonal group $G=SO(1,m)$ and its
maximal compact subgroup $H=SO(m)$, we obtain this action in the
explicit form (\ref{k30}).

Forthcoming Part II of our work is devoted to Lagrangians of these
matter fields and Higgs fields.

In accordance with Theorem \ref{redsub}, we consider a reduced
$H$-principal subbundles $P^h$ of $P$. Let
\mar{510f24}\beq
Y^h=(P^h\times V)/H \label{510f24}
\eeq
be an associated vector bundle with a typical fibre $V$ which
admits a group $H$ of exact symmetries. One can think of its
sections $s_h$ as describing matter fields in the presence of a
Higgs fields $h$ and some principal connection $A_h$ on $P^h$.

For different Higgs fields $h$, $h'$, the corresponding reduced
$H$-principle bundles $P^h$, $P^{h'}$ and, consequently, the
associated bundles $Y^h$, $Y^{h'}$ however fail to be isomorphic.

\begin{lemma} \label{k7} \mar{k7} If the quotient $G/H$ is
isomorphic to a Euclidean space $\mathbb R^m$, all $H$-principal
subbundles $P^h$ of a $G$-principal bundle $P$ are isomorphic to
each other \cite{ste}.
\end{lemma}

Nevertheless, if there exists an isomorphism between different
reduced subbundles $P^h$ and $P^{h'}$, this is an automorphism of
a $G$-principal bundle $P$ which sends $h$ to $h'$
\cite{book09,sard06a}. Therefore, a $V$-valued matter field can be
regarded only in a pair with a certain Higgs field $h$. A goal
thus is to describe the totality of these pairs $(s_h,h)$ for all
Higgs fields $h$.

For this purpose, let us consider a composite bundle
\mar{b3223a}\beq
P\to P/H\to X, \label{b3223a}
\eeq
where
\mar{b3194}\beq
 P_\Si=P\ar^{\pi_{P\Si}} P/H \label{b3194}
\eeq
is a principal bundle with a structure group $H$ and
\mar{b3193}\beq
\Si=P/H\ar^{\pi_{\Si X}} X \label{b3193}
\eeq
is a $P$-associated bundle with a typical fibre $G/H$ where a
structure group $G$ acts on the left.

Note that, given a global section $h$ of $\Si\to X$ (\ref{b3193}),
the corresponding reduced $H$-principal bundle $P^h$ is the
restriction $h^*P_\Si$ of the $H$-principal bundle $P_\Si$
(\ref{b3194}) to $h(X)\subset \Si$. Herewith, any atlas $\Psi_h$
of $P^h$ defined by a family of its local sections $\{U,z_h\}$
also is an atlas of a $G$-principal bundle $P$  and that of a
$P$-associated bundle $\Si\to X$ (\ref{b3193}) with $H$-valued
transition functions. With respect to this atlas $\Psi_h$ of
$\Si$, a global section $h$ of $\Si$ takes its values into the
center of the quotient $G/H$.

With the composite bundle (\ref{b3223a}), let us consider the
composite bundle
\mar{b3225}\beq
\pi_{YX}: Y\ar^{\pi_{Y\Si}} \Si\ar^{\pi_{\Si X}} X \label{b3225}
\eeq
where $Y\to \Si$ is a $P_\Si$-associated bundle
\mar{bnn}\beq
Y=(P\times V)/H \label{bnn}
\eeq
with a structure group $H$. Given a global section $h$ of the
fibre bundle $\Si\to X$ (\ref{b3193}) and the corresponding
reduced principal $H$ subbundle $P^h=h^*P$, the $P^h$-associated
fibre bundle (\ref{510f24}) is the restriction
\mar{b3226}\beq
Y^h=h^*Y=(h^*P\times V)/H \label{b3226}
\eeq
of a fibre bundle $Y\to\Si$ to $h(X)\subset \Si$.

As a consequence, every global section $s_h$ of the fibre bundle
$Y^h$ (\ref{b3226}) is a global section of the composite bundle
(\ref{b3225}) projected onto a section $h=\pi_{Y\Si}\circ s$ of a
fibre bundle $\Si\to X$. Conversely, every global section $s$ of
the composite bundle $Y\to X$ (\ref{b3225}) projected onto a
section $h=\pi_{Y\Si}\circ s$ of a fibre bundle $\Si\to X$ takes
its values into the subbundle $Y^h$ (\ref{b3226}). Hence, there is
one-to-one correspondence between the sections of the fibre bundle
$Y^h$ (\ref{510f24}) and the sections of the composite bundle
(\ref{b3225}) which cover $h$.

Thus, it is the composite bundle $Y\to X$ (\ref{b3225}) whose
sections describe the above mentioned totality of pairs $(s_h, h)$
of matter fields and Higgs fields in classical gauge theory with
spontaneous symmetry breaking \cite{book09,sard92,sard06a}. In
particular, one can show the following \cite{book09,sard06a}.

$\bullet$ An atlas $\{z_\Si\}$ of an $H$-principal bundle $H\to
\Si$ and, accordingly, of an associated bundle $Y\to\Si$ yields an
atlas $\{z_h=z_\Si\circ h\}$ of an $H$-principal bundle $P^h$ and,
consequently, $Y^h$.

$\bullet$ An $H$-principal connection $A_\Si$ on a fibre bundle
$Y\to \Si$ yields a pullback $H$-principal connection $A_h$ on
$Y^h$.

A key point under consideration here is the following.

\begin{theorem} \label{k12} \mar{k12}
The composite bundle $\pi_{YX}:Y\to X$ (\ref{b3225}) is a
$P$-associated bundle with a structure group $G$. Its typical
fibre is an $H$-principal bundle
\mar{wes}\beq
\pi_{WH}:W=(G\times V)/H \to G/H \label{wes}
\eeq
associated with an $H$-principal bundle
\mar{ggh}\beq
\pi_{GH}:G\to G/H. \label{ggh}
\eeq
\end{theorem}

\begin{proof}
Let us consider a principal bundle $P\to X$ as a $P$-associated
bundle
\be
P=(P\times G)/G, \qquad (pg',g)=(p,g'g), \qquad p\in P, \qquad
g,g'\in G,
\ee
whose typical fibre is a group space of $G$ which a group $G$ acts
on by left multiplications. Then the quotient (\ref{bnn}) can be
represented as
\be
&& Y=(P\times (G\times V)/H)/G, \\
&& (pg',(g\rho,v))= (pg',(g,\rho v))=(p,g'(g,\rho v))=(p,(g'g,\rho v)), \qquad \rho\in H.
\ee
It follows that $Y$ (\ref{bnn}) is a $P$-associated bundle with
the typical fibre $W$ (\ref{wes}) which the structure group $G$
acts on by the law
\mar{iik}\beq
g: (G\times V)/H \to (gG\times V)/H. \label{iik}
\eeq
This is a familiar induced representation of $G$ \cite{mack}. It
is an automorphism of the fibre bundle $W\to G/H$ (\ref{wes}).
Given an atlas $\Psi_{GH}=\{(U_\al,z_\al:U_\al\to G)\}$ of the
$H$-principal bundle $G\to G/H$ (\ref{ggh}), the induced
representation (\ref{iik}) reads
\mar{iik1}\ben
&& g: (\si,v)=(z_\al(\si),v)/H\to (\si',v')=(gz_\al(\si),v)/H= \nonumber \\
&& \qquad (z_\bt(\pi_{GH}(gz_\al(\si)))\rho,v)/H=
(z_\bt(\pi_{GH}(gz_\al(\si))),\rho v)/H,\qquad g\in G,\label{iik1}\\
&& \rho=z_\bt(\pi_{GH}(gz_\al(\si)))^{-1}gz_\al(\si)\in
H, \qquad \si\in U_\al, \qquad  \pi_{GH}(gz_\al(\si))\in U_\bt.
\nonumber
\een
If $H$ is a Cartan subgroup of $G$, an example of the induced
representation (\ref{iik1}) is the well-known non-linear
realizations \cite{col,book09,jos}.
\end{proof}

A problem however lies in the existence of an atlas of $Y$ both as
$P$-and $P_\Si$-associated bundles. Given an atlas
$\Psi=\{U_i,z_i:U_i\to P\}$ of $P$, we have the trivialization
charts
\mar{k1}\beq
 \psi_i: \pi_{YX}^{-1}(U_i)\to
U_i\times W \label{k1}
\eeq
of an associated bundle $Y\to X$. An atlas $\Psi_{GH}$ of $G\to
G/H$ in turn yields the trivialization charts
\mar{k2}\beq
\psi_\al: \pi_{WH}^{-1}(U_\al)\to U_\al\times V \label{k2}
\eeq
of the fibre bundle $W$ (\ref{wes}). Then the compositions of
trivialization morphisms (\ref{k1}) and (\ref{k2}) define fibred
coordinate charts
\mar{k3}\beq
(O; x^\mu,\si^m,v^A), \qquad O=\psi_i^{-1}(U_i\times
\pi_{WH}^{-1}(U_\al))\to U_i\times \pi_{WH}^{-1}(U_\al)\to
U_i\times U_\al\times V, \label{k3}
\eeq
of $Y\to \Si$ as a fibred manifold, but not a fibre bundle. They
possess the transition functions $(\si,v)\to
(\si'(x,\si),v'(x,\si,v))$ (\ref{iik1}).

If $G\to G/H$ and, consequently, $W\to G/H$ are trivial bundles,
then $U_\al=G/H$ and the charts (\ref{k3}) are fibre bundle charts
\be
&& O=\psi_i^{-1}(U_i\times W)=\pi_{YX}^{-1}(U_i)\to U_i\times W\to
U_i\times G/H\times V, \qquad O= \pi_{Y\Si}^{-1}(U_i\times G/H),
\ee
of $Y$ as a $P_\Si$-associated bundle.

\begin{lemma} \label{k8} \mar{k8}
By the well known theorem \cite{ste}, a fibre bundle $W\to G/H$
always is trivial over a Euclidean base $G/H$.
\end{lemma}

Any principal automorphism of a $G$-principal bundle $P\to X$,
being $G$-equivariant, also is $H$-equivariant and, thus, it is a
principal automorphism of a $H$-principal bundle $P\to\Si$.
Consequently, it yields an automorphism of the $P_\Si$-associated
bundle $Y$ (\ref{b3225}). Accordingly, every $G$-principal vector
field $\xi$ on $P\to X$ also is an $H$-principal vector field on
$P\to \Si$. It yield an infinitesimal gauge transformation
$\up_\xi$ of a composite bundle $Y$ seen as a $P$- and
$P_\Si$-associated bundle. It reads
\mar{rty}\beq
\up_\xi= \xi^\la\dr_\la + \xi^p(x^\m)J_p^m\dr_m +
\vt_\xi^a(x^\m,\si^k)I_a^A\dr_A, \label{rty}
\eeq
where $\{J_p\}$ is a representation of a Lie algebra $\cG$ of $G$
in $G/H$ and $\{I_a\}$ is a representation of a Lie algebra
$\mathfrak h$ of $H$ in $V$.

In view of Lemma \ref{k7} and Lemma \ref{k8}, we restrict our
consideration to groups $G$ and $H$ such that the quotient $G/H$
is an Euclidean space. As was mentioned above, this is the case of
a maximal compact subgroup $H \subset G$.

An additional condition on $G$ and $H$ is motivated by the
following fact \cite{book09}.

\begin{theorem} \label{k10} \mar{k10}
Let a Lie algebra $\cG$ of $G$ be the direct sum
\mar{g13}\beq
\cG = {\mathfrak h} \oplus {\mathfrak f} \label{g13}
\eeq
of a Lie algebra ${\mathfrak h}$ of $H$ and its supplement
$\mathfrak f$ obeying the commutation relations
\be
[{\mathfrak f},{\mathfrak f}]\subset {\mathfrak h}_r, \qquad
[{\mathfrak f},{\mathfrak h}_r]\subset \mathfrak f.
\ee
(e.g., $H$ is a Cartan subgroup of $G$). Let $A$ be a principal
connection on $P$. The ${\mathfrak h}$-valued component $A^h$ of
its pullback onto a reduced $H$-principal subbundle $P^h$ is a
principal connection on $P^h$. Moreover, there is a principal
connection on an $H$-principal bundle $P\to \Si$ so that its
restriction to very reduced $H$-principal subbundle $P^h$
coincides with $A^h$.
\end{theorem}

If the decomposition (\ref{g13}) holds, the well-known non-linear
realization of a Lie group $G$ possessing a subgroup $H$
exemplifies the induced representation (\ref{iik1})
\cite{col,jos}. In fact, it is a representation of the Lie algebra
of $G$ around its origin as follows \cite{book09,sard08}.

In this case, there exists an open neighbourhood $U$ of the unit
$\bb\in G$ such that any element $g\in U$ is uniquely brought into
the form
\be
g=\exp(F)\exp(I), \qquad F\in{\mathfrak f}, \qquad I\in{\mathfrak
h}_r.
\ee
Let $U_G$ be an open neighbourhood of the unit of $G$ such that
$U_G^2\subset U$, and let $U_0$ be an open neighbourhood of the
$H$-invariant center $\si_0$ of the quotient $G/H$ which consists
of elements
\be
\si=g\si_0=\exp(F)\si_0, \qquad g\in U_G.
\ee
Then there is a local section $z(g\si_0)=\exp(F)$ of $G\to G/H$
over $U_0$. With this local section, one can define the induced
representation (\ref{iik1}) of elements $g\in U_G\subset G$ on
$U_0\times V$ given by the expressions
\mar{g5a}\beq
g\exp(F)=\exp(F')\exp(I'), \qquad  g:(\exp(F),v)\to
(\exp(F'),\exp(I')v). \label{g5a}
\eeq
The corresponding representation of a Lie algebra $\cG$ of $G$
takes the following form. Let $\{F_\al\}$, $\{I_a\}$ be the bases
for $\mathfrak f$ and ${\mathfrak h}$, respectively. Their
elements obey the commutation relations
\be
[I_a,I_b]= c^d_{ab}I_d, \qquad [F_\al,F_\bt]= c^d_{\al\bt}I_d,
\qquad [F_\al,I_b]= c^\bt_{\al b}F_\bt.
\ee
Then the relation (\ref{g5a}) leads to the formulas
\mar{g7a,7',8a'}\ben
&& F_\al: F\to F'=F_\al
+\op\sum_{k=1}l_{2k}[\op\ldots_{2k}[F_\al,F],F],\ldots,F]-
\label{g7a}\\
&& \qquad  l_n\op\sum_{n=1}
 [\op\ldots_n[F,I'],I'],\ldots,I'],\nonumber\\
&& \qquad I'=
\op\sum_{k=1}l_{2k-1}[\op\ldots_{2k-1}[F_\al,F],F],\ldots,F],
\label{g7'}\\
&& I_a: F\to F'=
2\op\sum_{k=1}l_{2k-1}[\op\ldots_{2k-1}[I_a,F],F],\ldots,F],
\label{g8a}\\
&& \qquad I'=I_a, \label{g8a'}
\een
where coefficients $l_n$, $n=1,\ldots$, are obtained from the
recursion relation
\be
\frac{n}{(n+1)!}=\op\sum_{i=1}^n\frac{l_i}{(n+1-i)!}.
\ee
Let $U_F$ be an open subset of the origin of the vector space
$\mathfrak f$ such that the series (\ref{g7a}) -- (\ref{g8a'})
converge for all $F\in U_F$, $F_\al\in\mathfrak f$ and
$I_a\in{\mathfrak h}$. Then the above mentioned non-linear
realization of a Lie algebra $\cG$ in $U_F\times V$ reads
\be
F_\al: (F=\si^\bt F_\bt,v)\to (F'=\si'^\bt F_\bt,I'v), \qquad
I_a:(F=\si^\bt F_\bt,v)\to (F'=\si'^\bt F_\bt,I'v),
\ee
where $F'$ and $I'$ are given by the expressions (\ref{g7a}) --
(\ref{g8a}). In physical models, the coefficients $\si^\al$ of
$F=\si^\al F_\al$ are treated as Goldstone fields.

A problem is that the series (\ref{g7a}) -- (\ref{g8a'}) fail to
be summarized in general, and one usually restrict them to the
terms of second degree in $\si^\al$.

In the case of a pseudo-orthogonal group $G=SO(1,m)$ and its
maximal compact subgroup $SO(m)$, one however can bring the
expressions (\ref{g7a}) -- (\ref{g8a'}) into an explicit form. Let
us note that, as it was required above, the quotient
$SO(1,m)/SO(m)$ is homeomorphic to $\mathbb R^m$, and we have the
Lie algebra decomposition (\ref{g13}) such that
\mar{k20}\beq
so(1,m)=so(m)+ \mathfrak f. \label{k20}
\eeq
For instance, this is the case of a Lorentz group $SO(1,3)$ and
its subgroup $SO(3)$ of spatial rotations.

A key point is that there is a monomorphism of a Lie algebra
$so(1,m)$ to a real Clifford algebra $\mathcal{C\ell }(m)$
modelled over a pseudo-Euclidean space $(\mathbb R^m, \eta)$ with
a pseudo-Euclidean metric $\eta$ of signature $(+,-...-)$
\cite{law}.

Given the generating basis
\be
\{\g_i\},\qquad \g_i\g_k+\g_k\g_i=2\dl_{lk},
\ee
for a Clifford algebra $\mathcal{C\ell }(m)$, the above mentioned
monomorphism reads
\mar{k21}\beq
I_{0k}=\g_k, \qquad I_{ik}=\frac14[\g_k,\g_i], \label{k21}
\eeq
where $\{I_{0k},I_{ik}\}$ is a basis for a Lie algebra $so(1,m)$
and, accordingly, $\{I_{ik}\}$ is a basis for a subalgebra $so(m)$
and $\{F_k=\g_k=I_{0k}\}$ is a basis for its complement $\mathfrak
f$. Substituting these expressions into the formulas (\ref{g5a})
-- (\ref{g8a'}), we obtain
\ben
&& \exp(\si^k\g_k)= \cosh\si +\frac{\si^i}{\si}\sinh\si \g_i, \qquad
\si^2=\dl_{lk}\si^i\si^k, \nonumber\\
&& F_m=I_{0m}=u^k_m\frac{\dr}{\dr \si^k}
=\frac{\si^k\si^m}{\si^2}\left(1-\frac{2\si\cosh(2\si)}{\sinh\si}\right)\partial_k
+ \frac{2\si\cosh(2\si)}{\sinh(2\si)}\partial_m, \label{k30}\\
&& \qquad I'=\frac{2\si^k}{\si\tanh\si} I_{km},\nonumber \\
&& I_{ik}=\si^i\dr_k - \si^k\dr_i, \qquad I'=I_{ik}. \nonumber
\een

\end{document}